\documentclass[12pt,a4paper]{article}
\usepackage{times}
\title{\fontsize{18pt}{27pt}\selectfont
	{ 
		Study on Dynamical Behavior of Coinfection 
		\\Infectious Disease Model}}
\author{\fontsize{12pt}{18pt}\selectfont
	{
		Liu Yang}\\
	\fontsize{10.5pt}{15.75pt}\selectfont
	}
\date{}
\usepackage{amsmath,amsfonts,amssymb,amsthm}
\newtheorem{theorem}{Theorem}

\newtheorem{definition}{Definition}
\newtheorem{remark}{Remark}

\usepackage{graphicx}
\usepackage{subfigure}
\usepackage{float}
\usepackage[export]{adjustbox}
\usepackage{bibentry}
\usepackage[numbers]{natbib}
\usepackage{abstract}

\usepackage{xcolor}

\usepackage{url}
\usepackage{bm}
\usepackage{multirow}
\usepackage{booktabs}
\usepackage{epstopdf}
\usepackage{epsfig}
\usepackage{longtable}
\usepackage{supertabular}
\usepackage{algorithm}
\usepackage{algorithmic}
\usepackage{changepage}
\usepackage{enumerate}
\usepackage{caption}
\usepackage{indentfirst}
\usepackage[left=2.50cm,right=2.50cm,top=2.80cm,bottom=2.50cm]{geometry}
\usepackage{fancyhdr} 
\begin{document}
	\maketitle
	\lhead{}
	\chead{}
	\rhead{}
	\lfoot{}
	\cfoot{\thepage}
	\rfoot{}
	
	\begin{center}
		{\textbf{Abstract}}
	\end{center}
	\begin{adjustwidth}{1.06cm}{1.06cm}
		\hspace{1.5em}This paper conducts research on the established model and presents the
		main conclusions .  Firstly, by separately considering the infectivity of each of the two
		infectious diseases and the infectivity of the population simultaneously infected with the
		two infectious diseases, the existence of three types of boundary equilibrium points is
		determined, as well as the existence of the interior equilibrium point when the parameters
		are under specific conditions. Then, the stability of the equilibrium points is analyzed.
		It is concluded that under different parameter conditions, the stability of the disease
		free equilibrium point can exhibit various scenarios, such as a stable node or a saddle- node, etc. For the boundary equilibrium points, the situation is more intricate, and a
		cusp may occur. The stability of the interior equilibrium point under specific conditions
		is also presented. Finally, the degeneracy of the equilibrium points is studied through
		the bifurcation theory. Mainly, the saddle- node bifurcation occurring at the interior
		equilibrium point is obtained, and when the infection rate of the first infectious disease,
		the infection rate of the second infectious disease, and the infection rate of the co- infected
		population to other populations are selected as bifurcation parameters, a codimension- 3
		B- Tbifurcation is obtained. 
	\end{adjustwidth}
	\begin{adjustwidth}{1.06cm}{1.06cm}
		\fontsize{10.5pt}{15.75pt}\selectfont{\textbf{Keyword:} Infectious disease model;Coinfection;Equilibrium;Saddle-node bifurcation;
				B-T bifurcation}\\
	\end{adjustwidth}
	
	\section{Introduction}
	In Chapter 3, the SIS model was discussed and a mixed infection model was introduced. The mentioned mixed infection model ignores the situation where the diseased population simultaneously suffers from other infectious diseases. It mainly takes into account the infectious capabilities of different infectious diseases, and provides the basis for determining whether a single infectious disease will eventually become prevalent or not through the basic reproduction number. In this chapter, based on the SIS model, two different types of infectious diseases are considered. For the two types of infectious diseases that exhibit different transmission mechanisms, the situation is considered where the susceptible population, the population infected with the first type of infectious disease, the population infected with the second type of infectious disease, and the population that has contact with those infected with both infectious diseases become the population suffering from both infectious diseases simultaneously. And the phenomenon of whether the two infectious diseases will eventually form a common transmission is studied. The following model is constructed: 
	\begin{align}
		\label{SIS-Origin}
		\left\{\begin{aligned}
			\frac{\mathrm{d} S}{\mathrm{d} t} =& b - \alpha _{1}SI_{1}- \alpha _{2}SI_{2}-\alpha SI_{m}-bS+\theta_{1}I_1 \\
			\frac{\mathrm{d} I_{1} }{\mathrm{d} t } =&\alpha _{1}SI_{1}-\alpha_{2}I_{1}I_{2}-bI_{1}-\theta_{1}I_{1}-\alpha I_1I_{m}\\
			\frac{\mathrm{d} I_{2} }{\mathrm{d} t }=&\alpha _{2}SI_{2}-\alpha_{1}I_{1}I_{2}-bI_{2}-\alpha I_2I_{m}\\
			\frac{\mathrm{d} I_{m}}{\mathrm{d} t}=& (\alpha_1+\alpha_2)I_{1}I_{2}+\alpha I_{m}(S+I_1+I_2)-bI_{m}
		\end{aligned}\right.
	\end{align}
	Among them, $S$ represents the susceptible population; $I_1$ represents the population infected with the first type of infectious disease, and it is assumed that the first type of infectious disease cannot be completely cured, that is, the diseased population has the risk of being reinfected after being cured; $I_2$ represents the population infected with the second type of infectious disease; $I_{m}$ represents the population suffering from both infectious diseases simultaneously; $b$ represents the birth rate of the population, and to simplify the model, the population mortality rate is also set as $b$; $\alpha_i > 0, i = 1,2$ represents the infection rate of the $i$-th type of infectious disease; $\theta_{1} > 0$ represents the cure rate of the population infected with the first type of infectious disease; and the parameter $\alpha > 0$ represents the infection rate of the population suffering from both infectious diseases simultaneously to other populations.
	
	For the convenience of the research, the model is first simplified. Let $N(t)=S(t)+I_1(t)+I_2(t)+I_m(t)$ denote the total population density. By adding up the four equations in the system $\eqref{SIS-Origin}$, we can obtain
	\begin{align*}
		\frac{\mathrm{d}N}{\mathrm{d}t}=\frac{\mathrm{d} S}{\mathrm{d} t}+\frac{\mathrm{d} I_1}{\mathrm{d} t}+\frac{\mathrm{d} I_2}{\mathrm{d} t}+\frac{\mathrm{d} I_{m}}{\mathrm{d} t}=0
	\end{align*}
	Therefore, we can only discuss the system composed of the last three equations:
	\begin{align}
		\label{SIS-Second}
		\left\{\begin{aligned}
			\frac{\mathrm{d} I_{1} }{\mathrm{d} t } =&\alpha _{1}SI_{1}-\alpha_{2}I_{1}I_{2}-bI_{1}-\theta_{1}I_{1}-\alpha I_1I_{m}\\
			\frac{\mathrm{d} I_{2} }{\mathrm{d} t }=&\alpha _{2}SI_{2}-\alpha_{1}I_{1}I_{2}-bI_{2}-\alpha I_2I_{m}\\
			\frac{\mathrm{d} I_{m}}{\mathrm{d} t}=& (\alpha_1+\alpha_2)I_{1}I_{2}+\alpha I_{m}(S+I_1+I_2)-bI_{m}
		\end{aligned}\right.
	\end{align}
	
Then, let $S, I_1, I_2, I_{m}$ represent the densities of the corresponding populations relative to the total population $N$, and we can obtain $S + I_1+I_2 + I_{m}\equiv1$. Let $S = 1 - I_1 - I_2 - I_{m}$, and by making a variable substitution for the system (\ref{SIS-Second}) and a transformation of the time parameter $\tau = bt$, we can obtain the following system (here, for the convenience of representation, the time parameter is still denoted by $t$): 
	\begin{align}
		\label{SIS-True}
		\left\{\begin{aligned}
			\frac{\mathrm{d} I_{1} }{\mathrm{d} t } =&(a _{1}-r_1)I_{1}-a_1I_1^2-(a_{1}+a_{2})I_{1}I_{2}-(a_1+k)I_1I_{m}\\
			\frac{\mathrm{d} I_{2} }{\mathrm{d} t }=&(a _{2}-1)I_{2}-a_2I_2^2-(a_{1}+a_{2})I_{1}I_{2}-(a_2+k)I_2I_{m}\\
			\frac{\mathrm{d} I_{m}}{\mathrm{d} t}=& (a_1+a_2)I_{1}I_{2}+(k-1)I_{m}-kI_{m}^2
		\end{aligned}\right.
	\end{align}
	Wherein$a_i=\frac{\alpha_i}{b},i=1,2$,$k=\frac{\alpha}{b}$,$r_1=1+\frac{\theta_1}{b}$.
	
	In the subsequent content, the existence, stability and bifurcation of the equilibrium points will be studied. For the proofs in the subsequent content, we record here:
	\begin{align*}
		&u_1=a_1k+a_2^2,\\
		&u_2=a_2k+a_1^2,\\
		&u=a_1^2+a_2^2+a_1a_2,\\
		&v_1=a_1+a_2-a_2r_1-a_2^2,\\
		&v_2=a_1r_1+a_2r_1-a_1-a_1^2.
	\end{align*}
	\section{Existence and stability of equilibrium points}
	The study of equilibrium points plays a crucial role in the control and treatment of infectious diseases. Analyzing the existence of equilibrium points allows for a more straightforward and intuitive observation of the changes in the system. Analyzing the stability of equilibrium points enables us to determine the development trend of the prevalence of infectious diseases, helping us to analyze the relationship between the prevalence of infectious diseases and various factors. Based on this, effective measures can be taken to curb the spread of infectious diseases. To find the equilibrium points of the system (\ref{SIS-True}) is to solve the following system of ternary quadratic equations: 
	\begin{align}
		\label{rightEquation4SIS}
		\left\{\begin{aligned}
			&(a _{1}-r_1)I_{1}-a_1I_1^2-(a_{1}+a_{2})I_{1}I_{2}-(a_1+k)I_1I_{m}=0\\
			&(a _{2}-1)I_{2}-a_2I_2^2-(a_{1}+a_{2})I_{1}I_{2}-(a_2+k)I_2I_{m}=0\\
			&(a_1+a_2)I_{1}I_{2}+(k-1)I_{m}-kI_{m}^2=0
		\end{aligned}\right.
	\end{align}
	
	From the introduction of the two models in Chapter 3, it can be seen that for complex models, the basic reproduction number can still well represent the transmission ability of infectious diseases. Through the analysis of the model (\ref{SIS-True}), it is known that there always exists a disease-free equilibrium point $E_0=(0,0,0)$ in the system.
	
	By citing the definition of the basic reproduction number in the reference \cite{van2002}, and simultaneously considering the population flow situations of the three types of infected populations $I_1, I_2, I_m$, its reproduction matrix is constructed. The model is expressed as $\frac{\mathrm{d}X}{\mathrm{d}t}=\mathcal{F}(X)-\mathcal{V}(X)$, where $X=(I_1,I_2,I_m) \in \mathbb{R}^3$, and 
	\begin{align*}
		\mathcal{F}(X)&=\begin{pmatrix}
			a_1(1-I_1-I_2-I_m)I_1\\
			a_2(1-I_1-I_2-I_m)I_2\\
			k(1-I_m)I_m
		\end{pmatrix}   \\ \mathcal{V}(X)&=\begin{pmatrix}
			I_1(r_1+a_2I_2+kI_m)\\
			I_2(1+a_1I_1+kI_m)\\
			I_m[1-(a_1+a_2)I_1I_2]
		\end{pmatrix} 
	\end{align*}
	then the Jacobian matrices of $\mathcal{F}(X)$ and $\mathcal{V}(X)$ at the disease-free equilibrium point $E_0$ are: 
	\begin{align*}
		F=\frac{\partial \mathcal{F}}{\partial X} \Big|_{E_{0}}= \begin{pmatrix}
			a_1& 0 & 0\\
			0& a_2 & 0\\
			0& 0 & k
		\end{pmatrix}
		\qquad 
		V=\frac{\partial \mathcal{V}}{\partial X} \Big|_{E_{0}}= \begin{pmatrix}
			r_1& 0 & 0\\
			0& 1 & 0\\
			0& 0 & 1
		\end{pmatrix}
	\end{align*}
	so$FV^{-1}=\begin{pmatrix}
		\frac{a_1}{r_1} & 0 & 0\\
		0& a_2 & 0\\
		0& 0 & k
	\end{pmatrix}$,and$\rho(FV^{-1})=\max \{ \frac{a_1}{r_1},a_2,k\}$, the basic reproduction number of the system can be obtained.
	\begin{definition}
		The basic reproduction number of system(\ref{SIS-True}) is $R_0=\max \{ \frac{a_1}{r_1},a_2,k\}$.
	\end{definition}
	
	Here, without loss of generality, assume that the boundary equilibrium point corresponding to the first type of infectious disease appears first (the situation where the boundary equilibrium point corresponding to the second type of infectious disease appears first is similar to the discussion content below). Define the invasion reproduction number to represent the relative infectious ability of other infectious diseases when this boundary equilibrium point exists. According to the definition in the reference {\cite{gao2016coinfection}}, the expression of the invasion reproduction number of the second type of infectious disease for the equilibrium point of the first type of infectious disease is given by the reproduction matrix method: 
	\begin{align*}
		\mathcal{F}_2(I_2,I_m)&=\begin{pmatrix}
			\alpha_2SI_2\\
			( \alpha _1+\alpha_2)I_1I_2+\alpha I_m(S+I_1+I_2)
		\end{pmatrix},\\
		\mathcal{V}_2(I_2,I_m)&=\begin{pmatrix}
			\alpha_1I_1I_2+bI_2+\alpha I_2I_m \\
			bI_m
		\end{pmatrix}.
	\end{align*}
	and
	\begin{gather*}
		F_2=D\mathcal{F}_2(0,0)=\begin{pmatrix}
			\frac{a_2r_1}{a_1}  & 0\\
			\frac{(a_1+a_2)r_1}{a_1} & k
		\end{pmatrix},\\
		V_2=D\mathcal{V}_2(0,0)=\begin{pmatrix}
			a_1-r_1+1 & 0\\
			0 & 1
		\end{pmatrix}.
	\end{gather*}
	so
	\begin{align*}
		F_2V_2^{-1}=\begin{pmatrix}
			\frac{a_2r_1}{a_1(a_1-r_1+1)}  & 0\\
			\frac{(a_1+a_2)r_1}{a_1(a_1-r_1+1)} & k
		\end{pmatrix}.
	\end{align*}
	we can get that $\rho(F_2V_2^{-1})=\max\{\frac{a_2r_1}{a_1(a_1-r_1+1)},k \} $.
	\begin{definition}
		When the boundary equilibrium point corresponding to the first type of infectious disease appears first, the invasion reproduction number of the system is $R_2=\max\{\frac{a_2r_1}{a_1(a_1 - r_1 + 1)},k\}$.
	\end{definition}
	
	\begin{theorem}
		\label{existE123}
		(1) Regardless of the values of the parameters, the system always has a disease-free equilibrium point $E_0=(0,0,0)$;\\
		(2) When $R_{0,1}>1$, the system has a boundary equilibrium point $E_1 = (\frac{a_1 - r_1}{a_1},0,0)$;\\
		(3) When $R_{0,2}>1$, the system has a boundary equilibrium point $E_2=(0,\frac{a_2 - 1}{a_2},0)$;\\
		(4) When $k>1$, the system has a boundary equilibrium point $E_3=(0,0,\frac{k - 1}{k})$.
	\end{theorem}
	\begin{proof}
		Assume that the boundary equilibrium point corresponding to the first type of infectious disease has the form $E_1=(I_1,0,0)$ where $I_1\neq0$. Then, substituting it into the system of equations (\ref{rightEquation4SIS}), the last two equations satisfy that both the left and right sides are 0. For the first equation, we have:
		\begin{align*}
			(a_1 - r_1)-a_1I_1 = 0,
		\end{align*}
		Therefore, when $R_{0,1}>1$, we can obtain $I_1=\frac{a_1 - r_1}{a_1}>0$. That is, at this time, the system has a boundary equilibrium point $E_1 = (\frac{a_1 - r_1}{a_1},0,0)$. When $R_{0,2}>1$ and $k>1$, the corresponding boundary equilibrium points $E_2$ and $E_3$ can be obtained in a similar way.
	\end{proof}
	\begin{theorem}
		\label{criticalPointMix}
		Consider two types of special boundary equilibrium points: \\
		(1) When $k > 1$ and $a_1>r_1k+(k - 1)k$, the system has a boundary equilibrium point $E_{1m}=(\frac{a_1 - k^2-(r_1 - 1)k}{a_1k},0,\frac{k - 1}{k})$; \\
		(2) When $k > 1$ and $a_2>k^2$, the system has a boundary equilibrium point $E_{2m}=(0,\frac{a_2 - k^2}{a_2k},\frac{k - 1}{k})$.
	\end{theorem}
	\begin{proof}
		For the system of equations (\ref{rightEquation4SIS}), consider the special boundary equilibrium points. When $I_m = 0$ and $I_1I_2\neq0$, it is obvious that the third equation of the system of equations is not satisfied. Therefore, there is no such type of boundary equilibrium point;
		
		When $I_2 = 0$ and $I_1I_m\neq0$, consider the following system of equations:
		\begin{align}
			\left\{\begin{aligned}
				&(a _{1}-r_1)-a_1I_1-(a_1 + k)I_{m}=0\\
				&(k - 1)-kI_{m}=0
			\end{aligned}\right.
		\end{align}
		When $k>1$ and $a_1>r_1k+(k - 1)k$, by solving the equations, we can obtain
		\begin{align*}
			I_1=\frac{a_1 - k^2-(r_1 - 1)k}{a_1k},\quad I_m=\frac{k - 1}{k}.
		\end{align*}
		Therefore, at this time, the boundary equilibrium point $E_{1m}$ exists.
		
		When $I_1 = 0$ and $I_2I_m\neq0$, consider the following system of equations:
		\begin{align}
			\left\{\begin{aligned}
				&(a _{2}-1)-a_2I_2-(a_2 + k)I_{m}=0\\
				&(k - 1)-kI_{m}=0
			\end{aligned}\right.
		\end{align}
		When $k>1$ and $a_2>k^2$, by solving the equations, we can obtain
		\begin{align*}
			I_2=\frac{a_2 - k^2}{a_2k},\quad I_m=\frac{k - 1}{k}.
		\end{align*}
		Therefore, at this time, the boundary equilibrium point $E_{2m}$ exists.
	\end{proof}
	\begin{remark}
		By comparing Theorem \ref{criticalPointMix} and Theorem \ref{existE123}, it can be found that when $k > 1$, the component values corresponding to $I_m$ of the boundary equilibrium points $E_3$, $E_{1m}$ and $E_{2m}$ are equal. At this time, if it is desired that the values of the boundary equilibrium points on the $I_1$ or $I_2$ components are not zero, then the threshold of the infection rate of the corresponding infectious disease is higher. That is, the corresponding infectious disease needs to exhibit a stronger infectious ability at this time. Therefore, on the premise that the total population density remains unchanged, the spread of co-infection will also show a certain competitive relationship with the spread of the first type of infectious disease or the second type of infectious disease.
	\end{remark}
	\begin{theorem}
		\label{thmNopoint}
		When the basic reproduction number $R_0 < 1$, there are no other equilibrium points in the system.
	\end{theorem}
	\begin{proof}
		When the basic reproduction number $R_0 < 1$, we have $a_1 < r_1$, $a_2 < 1$, and $k < 1$. Obviously, the boundary equilibrium points $E_1$, $E_2$, $E_3$, $E_{1m}$, and $E_{2m}$ do not exist. Next, we consider the existence of the interior equilibrium point when $R_0 < 1$.
		
		For the equation (\ref{rightEquation4SIS}), when $I_1$, $I_2$, and $I_m$ are all non - zero, we solve it and get:
		\begin{align}
			\label{eq4fixpoints}
			\left\{\begin{aligned}
				&a_{1}-r_1-a_1I_1-(a_{1}+a_{2})I_{2}-(a_1 + k)I_{m}=0\\
				&a_{2}-1-a_2I_2-(a_{1}+a_{2})I_{1}-(a_2 + k)I_{m}=0\\
				&(a_1 + a_2)I_{1}I_{2}+(k - 1)I_{m}-kI_{m}^2=0
			\end{aligned}\right.
		\end{align}
		Taking $I_m$ as a parameter, the system of linear equations in two variables about $I_1$ and $I_2$ is
		\begin{align*}
			\left\{\begin{aligned}
				&a_1I_1+(a_{1}+a_{2})I_{2}=a_{1}-r_1-(a_1 + k)I_{m}\\
				&(a_{1}+a_{2})I_{1}+a_2I_2=a_{2}-1-(a_2 + k)I_{m}\\
			\end{aligned}\right.
		\end{align*}
		We solve the above system of linear equations in two variables by Cramer's rule. We get
		\begin{align}
			\label{I1}
			I_1=-\frac{u_1I_m + v_1}{u},\\
			\label{I2}
			I_2=-\frac{u_2I_m + v_2}{u}.
		\end{align}
		
		Next, we explain the relationship between the signs of $v_1$ and $v_2$ and the existence of the interior equilibrium point. For
		\begin{align*}
			v_1=a_1 + a_2-a_2r_1-a_2^2,\quad
			v_2=a_1r_1 + a_2r_1-a_1-a_1^2.
		\end{align*}
		
		When $v_1 = 0$ and $v_2 = 0$, we have
		\begin{gather*}
			a_1=a_2^2+(r_1 - 1)a_2,\quad
			a_2=\frac{1}{r_1}[a_1^2+(1 - r_1)a_1].
		\end{gather*}
		Regarding $r_1$ as a parameter and $a_1$, $a_2$ as variables, the graphs of the two curves corresponding to $v_1 = 0$ and $v_2 = 0$ in the two - dimensional plane are as follows:
		\begin{figure}[H]
			\label{fig4a1a2}
			\centering
			\includegraphics[scale=0.6]{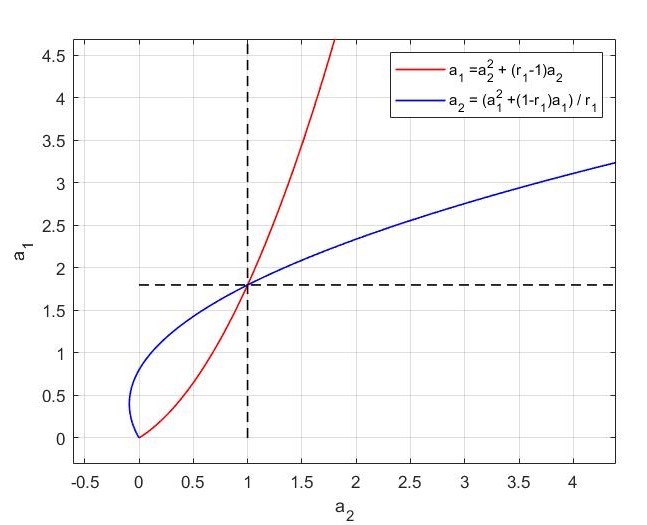}
			\caption{The graph of parameter $a_1$ with respect to $a_2$}
			\label{figure}
		\end{figure}
		Then, from the graph, we can see that the two curves intersect at the point $(r_1,1)$. When $R_{0,1}<1$, there are only the following cases:
		
		(1) $v_1\geq0, v_2\geq0$, (2) $v_1 > 0, v_2<0$, (3) $v_1<0, v_2>0$.
		
		Let $v_1 = 0$. Suppose that there is a positive solution $I_m^*$ for the equation (\ref{eq4Im}) at this time. Substituting it into the expression of $I_1$ with respect to $I_m$ (\ref{I1}), we get $I_1^*=-\frac{u_1I_m^*}{u}$. Also, it is obvious that the parameters $u_1>0, u>0$, then $I_1^*<0$. Therefore, $I_m^*$ is not a solution corresponding to the interior equilibrium point.
		Then, consider the case when there is a positive solution $I_m^*$ when $v_1>0$. Substituting it into the formula (\ref{I1}), we can also get $I_1 = -\frac{u_1I_m^*+v_1}{u}<0$ in the same way. Therefore, when $v_1\geq0$, the system (\ref{SIS-True}) has no interior equilibrium point. Similarly, when $v_2\geq0$, the system also has no interior equilibrium point.
		
		Based on the above discussion, it is obvious that when $R_{0,1}<1$, the system has no interior equilibrium point. Similarly, when $R_{0,2}<1$, the system also has no interior equilibrium point. Furthermore, when the basic reproduction number $R_0 < 1$, the system has no interior equilibrium point at all.
	\end{proof}
	\begin{remark}
		According to the theorem, when the basic reproduction numbers $R_{0,j}<1,j = 1,2$ corresponding to various types of infectious diseases, that is, when there are neither the boundary equilibrium point corresponding to the first type of infectious disease nor the boundary equilibrium point corresponding to the second type of infectious disease in the system, there is also no interior equilibrium point in the system. That is, when neither of the two infectious diseases can be prevalent, the situation of co-infection prevalence cannot occur either.
	\end{remark}
	\noindent Note:
	\begin{align*}
		k_0=\frac{1}{2}\left( a_1+a_2+\frac{a_1a_2}{a_1+a_2}-\sqrt{a_1^2+a_2^2+\frac{a_1^2a_2^2}{(a_1+a_2)^2}}  \right) \\ k_1=\frac{1}{2}\left( a_1+a_2+\frac{a_1a_2}{a_1+a_2}+\sqrt{a_1^2+a_2^2+\frac{a_1^2a_2^2}{(a_1+a_2)^2}}\right) 
	\end{align*}
	\begin{theorem}\label{exist4E4E5}
		When $a_1>r_1$, $a_2 > 1$ and $a_1\in\left(\frac{1 - r_1}{2}+\sqrt{a_2r_1+\frac{(1 - r_1)^2}{4}},a_2^2+(r_1 - 1)a_2\right)$,\\
		(1) If $k = k_0\leq1$, then the system has an interior equilibrium point $E_4$ at this time;\\
		(2) If $k\in(k_0,k_1)$, then the system has an interior equilibrium point $E_5$ at this time.
	\end{theorem}
	\begin{proof}
		Substitute the expressions of $I_1$ and $I_2$ in terms of $I_m$ into the third equation of the system (\ref{eq4fixpoints}), we can get:
		\begin{align} 
			\label{eq4Im}
			aI_m^2 + bI_m + c = 0
		\end{align}
		where
		\begin{align*}
			&a=\frac{u_1u_2(a_1 + a_2)}{u^2}-k,\\
			&b=k - 1+\frac{(a_1 + a_2)(u_1v_2+u_2v_1)}{u^2},\\
			&c=\frac{v_1v_2(a_1 + a_2)}{u^2}.
		\end{align*}
		
		For the univariate quadratic equation (\ref{eq4Im}), consider the positive roots of the equation when $ac<0$.
		
		From the proof content of Theorem (\ref{thmNopoint}), it is known that the system may have an interior equilibrium point only when $R_0>1$ and $v_1<0$, $v_2<0$. Then when $a_1\in\left(\frac{1 - r_1}{2}+\sqrt{a_2r_1+\frac{(1 - r_1)^2}{4}},a_2^2+(r_1 - 1)a_2\right)$, we have $v_1<0$, $v_2<0$, $v_1v_2>0$. Therefore, the coefficient of the constant term $c>0$.
		
		Then discuss the sign of the coefficient of the quadratic term $a$. Substitute the expressions of $u_1$ and $u_2$ in terms of $a_1$, $a_2$ and $k$ into $ku^2=u_1u_2(a_1 + a_2)$, we can get:
		\begin{align*}
			\label{eq4k}
			a_1a_2k^2+(a_1^3+a_2^3-\frac{u^2}{a_1 + a_2})k+a_1^2a_2^2 = 0 
		\end{align*}
		Discuss the properties of the univariate quadratic equation about $k$:
		\begin{align}
			\begin{split}
				\Delta=&\left(a_1^3+a_2^3-\frac{u^2}{a_1 + a_2}\right)^2-4a_1^3a_2^3\\
				=&(a_1^3+a_2^3)^2-2(a_1^2-a_1a_2+a_2^2)[(a_1 + a_2)^4-2a_1a_2(a_1 + a_2)^2+a_1^2a_2^2]\\
				&+(a_1 + a_2)^6-4(a_1 + a_2)^4a_1a_2+6(a_1 + a_2)^2a_1^2a_2^2 - 4a_1^3a_2^3-4a_1^3a_2^3\\
				&+\frac{a_1^4a_2^4}{(a_1 + a_2)^2}\\
				=&a_1^4a_2^2+a_1^2a_2^4+\frac{a_1^4a_2^4}{(a_1 + a_2)^2}>0
			\end{split}
		\end{align}
		Solve the univariate quadratic equation about $k$ to get the zero solutions $k_0$ and $k_1$.
		
		(1) When $k = k_0\leq1$, the coefficient of the quadratic term $a = 0$, the coefficient $b<0$ and $c>0$. Then the equation (\ref{eq4Im}) has a positive root $I_m=-\frac{c}{b}$. Substitute it back into the expressions of $I_m$ in equations \eqref{I1} and \eqref{I2} to get the interior equilibrium point $E_4$.
		
		(2) When $k\in(k_0,k_1)$, the coefficient $a<0$. At this time, the discriminant of the roots of the equation \eqref{eq4Im} $b^2 - 4ac>0$. Therefore, the equation has a positive root $I_m=\frac{-b-\sqrt{b^2 - 4ac}}{2a}$. Substitute it back into the expressions of $I_m$ in equations \eqref{I1} and \eqref{I2} to get the interior equilibrium point $E_5$.
	\end{proof}
	
	Consider the stability of the boundary equilibrium points and the interior equilibrium points of the model. The Jacobian matrix of the right-hand side equations of the system (\ref{SIS-True}) at the equilibrium point $(I_1, I_2, I_m)$ is: 
	\begin{align*}
		J=\begin{pmatrix}
			h_1(I_1,I_2,I_m)& -(a_1+a_2)I_1 & -(a_1+k)I_1\\
			-(a_1+a_2)I_2& h_2(I_1,I_2,I_m) & -(a_2+k)I_2\\
			(a_1+a_2)I_2& (a_1+a_2)I_1 &k-1-2kI_m
		\end{pmatrix}
	\end{align*}
	wherein:
	\begin{align*}
		h_1(I_1,I_2,I_m)&=a_1-r_1-2a_1I_1-(a_1+a_2)I_2-(a_1+k)I_m\\
		h_2(I_1,I_2,I_m)&=a_2-1-2a_2I_2-(a_1+a_2)I_1-(a_2+k)I_m
	\end{align*}
	\begin{theorem}
		For the disease-free equilibrium point $E_0$,\\
		(1) When $R_0 < 1$, this equilibrium point is a stable node;\\
		(2) When $R_0>1$, here we assume that $a_1 > r_1$,\\
		(i) If $a_2 < 1$ and $k < 1$, then the disease-free equilibrium point $E_0$ is a saddle point;\\
		(ii) If $a_2 = 1$ and $k\neq1$ or $k = 1$ and $a_2\neq1$, the disease-free equilibrium point $E_0$ is a saddle-node;\\
		(iii) If $a_2>1$ and $k > 1$, then the disease-free equilibrium point $E_0$ is an unstable node;
	\end{theorem}
	\begin{proof}
		For the disease-free equilibrium point $E_0=(0,0,0)$, the corresponding Jacobian matrix is:
		\begin{align*}
			J_0=\begin{pmatrix}
				a_1 - r_1&0&0\\
				0&a_2 - 1&0\\
				0&0&k - 1
			\end{pmatrix}
		\end{align*}
		The corresponding eigenvalues are $\lambda_{1}=a_1 - r_1$, $\lambda_{2}=a_2 - 1$, $\lambda_{3}=k - 1$.
		
		(1) According to the reference \cite{guckenheimer2013nonlinear}, when $R_0 < 1$, the disease-free equilibrium point $E_0$ is a stable node; and from Theorem \ref{thmNopoint}, there are no other equilibrium points at this time. Therefore, the disease-free equilibrium point $E_0$ is globally stable at this time;
		
		(2) When $a_1>r_1$ and $a_2 < 1$, $k < 1$, the eigenvalues have different signs, and the equilibrium point is a saddle point; when $a_1>r_1$ and $a_2 = 1$, $k\neq1$, we have $\lambda_{1}=a_1 - r_1>0$, $\lambda_{2}=a_2 - 1 = 0$, $\lambda_{3}=k - 1\neq0$. At this time, it is obvious that the disease-free equilibrium point is a Lyapunov-type singular point .
		
		At this time, by the center manifold theorem, for $I_1$ and $I_m$, we have $I_1 = O(I_2^2)$ and $I_m = O(I_2^2)$. Substituting it back into the equation of the original system about $I_2$ gives:
		\begin{align*}
			\dot{I}_2=-I_2^2 + O(I_2^3),
		\end{align*}
		Therefore, the disease-free equilibrium point is a saddle-node at this time.
		
		When $a_1>r_1$, $a_2>1$ and $k > 1$, it is obvious that the disease-free equilibrium point $E_0$ is an unstable node.
	\end{proof}
	
	It can be seen from the theorem that when the disease infection rate is less than the corresponding treatment rate, the disease-free equilibrium point is globally stable. Therefore, for infectious diseases, it is necessary to achieve timely control and keep the infection rate at a low level. At the same time, it is also necessary to strengthen the prevention of infectious diseases, improve the response capabilities of hospitals in various places, and avoid sudden outbreaks.
	\begin{theorem}\label{stablethmE1}
		For the boundary equilibrium point $E_1$,\\
		(1) When the invasion reproduction number $R_2 < 1$, the equilibrium point $E_1$ is a locally asymptotically stable node;\\
		(2) If $a_2=\frac{a_1(a_1 - r_1 + 1)}{r_1}$ and $k\neq1$, or $k = 1$ and $a_2\neq\frac{a_1(a_1 - r_1 + 1)}{r_1}$, then the equilibrium point $E_1$ is a saddle-node;\\
		(3) If $a_2=\frac{a_1(a_1 - r_1 + 1)}{r_1}$ and $k = 1$, then the equilibrium point $E_1$ is a cusp of codimension 3.
	\end{theorem}
	\begin{proof}
		For the equilibrium point $E_1 = (\frac{a_1 - r_1}{a_1}, 0, 0)$, the Jacobian matrix of the system at this point is
		\begin{align*}
			J_1=\begin{pmatrix}
				r_1 - a_1&-\frac{(a_1 + a_2)(a_1 - r_1)}{a_1}&-\frac{(a_1 + k)(a_1 - r_1)}{a_1}\\
				0&a_2 - 1-\frac{(a_1 + a_2)(a_1 - r_1)}{a_1}&0\\
				0&\frac{(a_1 + a_2)(a_1 - r_1)}{a_1}&k - 1
			\end{pmatrix}
		\end{align*}
		Then the corresponding eigenvalues are $\lambda_{1}=r_1 - a_1<0$, $\lambda_{2}=a_2 - 1-\frac{(a_1 + a_2)(a_1 - r_1)}{a_1}$, $\lambda_{3}=k - 1$.
		(1) When the invasion reproduction number $R_2 < 1$ and $k < 1$, we have $a_2<\frac{a_1(a_1 - r_1 + 1)}{r_1}$. Then at this time, $\lambda_2=a_2 - 1-\frac{(a_1 + a_2)(a_1 - r_1)}{a_1}<0$ and $\lambda_{3}=k - 1<0$. Obviously, the equilibrium point $E_1$ is a locally asymptotically stable node at this time.
		(2) When $k = 1$, the matrix has only one zero eigenvalue at this time. Denote $J_1(k = 1)=A_1$. By the coordinate transformation $(x,y,z)=(I_1-\frac{a_1 - r_1}{a_1}, 0, 0)$, the equilibrium point $E_1$ at this time is translated to the origin, and we can get:
		\begin{align}
			\label{E1SN}
			\frac{\mathrm{d}X}{\mathrm{d}t}=A_1X + F(X)
		\end{align}
		where $X=(x,y,z)$, and $F(X)$ has the following form:
		\begin{align}
			F(X)=\begin{pmatrix}
				-a_1x^2-(a_1 + a_2)xy-(a_1 + 1)xz\\
				-a_2y^2-(a_1 + a_2)xy-(a_2 + 1)yz\\
				(a_1 + a_2)xy - z^2
			\end{pmatrix}
		\end{align}
		
		Find the eigenvectors $Aq = 0$ and $A^Tp = 0$, and we can get:
		\begin{align*}
			q=\begin{pmatrix}
				\frac{a_1 + 1}{a_1}\\
				0\\
				1
			\end{pmatrix},\quad
			p=\begin{pmatrix}
				0\\
				\frac{\gamma}{a_2 - 1-\gamma}\\
				1
			\end{pmatrix}
		\end{align*}
		where $\gamma=\frac{(a_1 + a_2)(a_1 - r_1)}{a_1}$. According to the reference {\cite{kuznetsov1998elements}}, the center manifold corresponding to the saddle-node bifurcation has the following form:
		\begin{align*}
			\frac{\mathrm{d}w}{\mathrm{d}t}=\sigma w^2+\theta w^3+O(w^4).
		\end{align*}
		where $w = \langle p,X\rangle=\frac{\gamma}{a_2 - 1-\gamma}y + z$. Substitute it into the formula:
		\begin{align*}
			\sigma=-1<0,\quad
			\theta=\frac{a_1 + 2}{2a_1(r_1 - a_1)}<0.
		\end{align*}
		Therefore, the system (\ref{E1SN}) is equivalent to the system near the origin:
		\begin{align*}
			\frac{\mathrm{d}w}{\mathrm{d}t}=-w^2+O(w^3).
		\end{align*}
		Therefore, it can be known that the equilibrium point $E_1$ is a saddle-node of codimension 1 at this time.
		
		(3) When $a_2 = \frac{a_1(a_1 - r_1 + 1)}{r_1}$ and $k = 1$, the eigenvalues $\lambda_1 = r_1 - a_1<0$, $\lambda_2=\lambda_3 = 0$. By performing a coordinate transformation on the system, we translate the equilibrium point $E_1$ to the origin, and obtain:
		\begin{align*}
			\begin{cases}
				\frac{\mathrm{d}x}{\mathrm{d}t}&=(r_1 - a_1)x-\frac{(a_1 + 1)(a_1 - r_1)}{r_1}y-\frac{(a_1 + 1)(a_1 - r_1)}{a_1}z - a_1x^2-\frac{a_1(a_1 + 1)}{r_1}xy\\
				\frac{\mathrm{d}y}{\mathrm{d}t}&=-\frac{a_1(a_1 + 1 - r_1)}{r_1}y^2-\frac{a_1(a_1 + 1)}{r_1}xy-\frac{a_1(a_1 + 1 - r_1)+r_1}{r_1}yz\\
				\frac{\mathrm{d}z}{\mathrm{d}t}&=\frac{(a_1 + 1)(a_1 - r_1)}{r_1}y - z^2
			\end{cases}
		\end{align*}
		According to the center manifold theorem, there exists a center manifold $x = u_{10}y+u_{01}z + u_{20}y^2+u_{11}yz + u_{02}z^2+O(\|(y,z)\|^3)$. Substituting it into the translated system, we can get:
		\begin{align*}
			\frac{\mathrm{d}x}{\mathrm{d}t}&=(r_1 - a_1)(u_{10}y+u_{01}z + u_{20}y^2+u_{11}yz + u_{02}z^2)-\frac{(a_1 + 1)(a_1 - r_1)}{r_1}y\\
			&-\frac{(a_1 + 1)(a_1 - r_1)}{a_1}z - a_1(u_{10}y+u_{01}z + u_{20}y^2+u_{11}yz + u_{02}z^2)^2\\
			&-\frac{a_1(a_1 + 1)}{r_1}(u_{10}y+u_{01}z + u_{20}y^2+u_{11}yz + u_{02}z^2)y\\
			&-(a_1 + 1)(u_{10}y+u_{01}z + u_{20}y^2+u_{11}yz + u_{02}z^2)z\\
			&=(u_{10}(r_1 - a_1)-\frac{(a_1 + 1)(a_1 - r_1)}{r_1})y+(u_{01}(r_1 - a_1)-\frac{(a_1 + 1)(a_1 - r_1)}{a_1})z\\
			&+(u_{20}(r_1 - a_1)-u_{10}^2a_1 - u_{10}\frac{a_1(a_1 + 1)}{r_1})y^2\\
			&+(u_{11}(r_1 - a_1)-2u_{10}u_{01}a_1 - u_{01}\frac{a_1(a_1 + 1)}{r_1}-u_{10}(a_1 + 1))\\
			&+z^2(u_{02}(r_1 - a_1)-u_{01}^2a_1 - u_{01}(a_1 + 1))+\cdots
		\end{align*}
		At the same time, differentiating both sides of the center manifold $x = u_{10}y+u_{01}z + u_{20}y^2+u_{11}yz + u_{02}z^2+O(\|(y,z)\|^3)$ with respect to time, we can obtain:
		\begin{align*}
			\frac{\mathrm{d}x}{\mathrm{d}t}&=u_{10}\frac{\mathrm{d}y}{\mathrm{d}t}+u_{01}\frac{\mathrm{d}z}{\mathrm{d}t}+2u_{20}y\frac{\mathrm{d}y}{\mathrm{d}t}+u_{11}z\frac{\mathrm{d}y}{\mathrm{d}t}+u_{11}y\frac{\mathrm{d}z}{\mathrm{d}t}+2u_{02}z\frac{\mathrm{d}z}{\mathrm{d}t}\\
			&=\frac{u_{01}(a_1 + 1)(a_1 - r_1)}{r_1}y\\
			&+(-\frac{u_{10}a_1(a_1 + 1 - r_1)}{r_1}+\frac{u_{11}(a_1 + 1)(a_1 - r_1)}{r_1}-\frac{u_{10}^2a_1(a_1 + 1)}{r_1}\\
			&+\frac{u_{10}u_{01}a_1(a_1 + 1)}{r_1})y^2+(-\frac{u_{10}a_1(a_1 + 1 - r_1)+r_1}{r_1}u_{10}-\frac{u_{10}u_{01}a_1(a_1 + 1)}{r_1}\\
			&+\frac{2u_{02}(a_1 + 1)(a_1 - r_1)}{r_1}+\frac{u_{01}^2a_1(a_1 + 1)}{r_1})yz - u_{01}z^2+\cdots
		\end{align*}
		By comparing the coefficients, we can get:
		\begin{align*}
			u_{01}&=-\frac{a_1 + 1}{a_1}, \quad
			u_{10}=\frac{a_1 + 1}{a_1r_1}, \\
			u_{02}&=-\frac{a_1 + 1}{a_1(a_1 - r_1)}, \quad
			u_{11}=-\frac{(a_1^3 - a_1 - 2)r_1+(a_1 + 1)^2}{(a_1 - r_1)a_1r_1^2}, \\
			u_{20}&=\frac{3a_1^2r_1 + 4a_1r_1+2r_1+a_1^2r_1^2+a_1r_1^2-a_1^4r_1-2a_1^3-6a_1^2-6a_1 - 2}{a_1r_1^3(a_1 - r_1)}.
		\end{align*}
		Substituting the obtained center manifold into the translated system, we can get:
		\begin{align*}
			\begin{cases}
				\frac{\mathrm{d}y}{\mathrm{d}t}=a_{20}y^2+a_{11}yz+a_{30}y^3+a_{21}y^2z+a_{12}yz^2+O(\|(y,z)\|^4) \\
				\frac{\mathrm{d}z}{\mathrm{d}t}=b_{10}y+b_{20}y^2+b_{11}yz+b_{30}y^3+b_{21}y^2z+b_{12}yz^2+O(\|(y,z)\|^4)
			\end{cases}
		\end{align*}
		where:
		\begin{align*}
			a_{20}&=-\frac{a_1(a_1 - r_1 + 1)}{r_1}+\frac{a_1(a_1 + 1)u_{10}}{r_1}\\
			a_{11}&=-\frac{a_1(a_1 - r_1 + 1)}{r_1}+\frac{a_1(a_1 + 1)u_{01}-r_1}{r_1}\\
			b_{10}&=\frac{a_1(a_1 - r_1 + 1)-r_1}{r_1}
		\end{align*}
		Denote $l_{00}=-\frac{a_1(a_1 + 1)}{r_1}$, then the other coefficients have the following form:
		\begin{align*}
			a_{ij}=l_{00}u_{i - 1j},\quad b_{ij}=l_{00}u_{i - 1j},i = 1,2,4,j=0,1,2.
		\end{align*}
		
		Then perform a coordinate transformation, let:
		\begin{align*}
			\begin{cases}
				x_1=z \\
				x_2=b_{10}y
			\end{cases}
		\end{align*}
		so we can get
		\begin{align*}
			\begin{cases}
				\frac{\mathrm{d}x_1}{\mathrm{d}t}=x_2+\frac{b_{11}}{b_{10}}x_1x_2+\frac{b_{20}}{b_{10}^2}x_2^2+\frac{b_{12}}{b_{10}}x_1^2x_2+\frac{b_{21}}{b_{10}^2}x_1x_2^2+\frac{b_{30}}{b_{10}^3}x_1^3+O(\|x\|^4)\\
				\frac{\mathrm{d}x_2}{\mathrm{d}t}=a_{11}x_1x_2+\frac{a_{20}}{b_{10}}x_2^2+\frac{a_{21}}{b_{10}}x_1x_2^2+a_{12}x_1^2x_2+\frac{a_{30}}{b_{10}^2}x_2^3+O(\|x\|^4)
			\end{cases}
		\end{align*}
		
		\begin{align*}
			\begin{cases}
				x_3=x_1 \\
				x_4=x_2+\frac{b_{11}}{b_{10}}x_1x_2+\frac{b_{20}}{b_{10}^2}x_2^2+\frac{b_{12}}{b_{10}}x_1^2x_2+\frac{b_{21}}{b_{10}^2}x_1x_2^2+\frac{b_{30}}{b_{10}^3}x_1^3+O(\|x\|^4)
			\end{cases}
		\end{align*}
		
		\begin{align*}
			\begin{cases}
				x_1=x_3 \\
				x_2=x_4+v_{11}x_3x_4+v_{02}x_4^2+v_{21}x_3^2x_4+v_{12}x_3x_4^2+v_{03}x_4^3+O(\|x\|^4)
			\end{cases}
		\end{align*}
		wherein:
		\begin{align*}
			v_{11}&=-\frac{b_{11}}{b_{10}},v_{02}=-\frac{b_{20}}{b_{10}^2}, \\
			v_{21}&=\frac{b_{11}^2-b_{12}b_{10}}{b_{10}^2}.\\
			v_{12}&=\frac{3b_{11}b_{20}-b_{21}b_{10}}{b_{10}^3},v_{03}=\frac{2b_{20}^2}{b_{10}^4}.
		\end{align*}
		Substitute the new variables \((x_3, x_4)\) into the original system, and we can get:
		\begin{align}
			\label{BT-1}
			\begin{cases}
				\frac{\mathrm{d}x_3}{\mathrm{d}t} = x_4\\
				\frac{\mathrm{d}x_4}{\mathrm{d}t} = e_{11}x_3x_4+e_{02}x_4^2  + e_{21}x_3^2x_4 + e_{12}x_3x_4^2+e_{03}x_4^3+O(\|x\|^4)\\
			\end{cases}
		\end{align}
		wherein:
		\begin{align*}
			e_{11}&=a_{11}, \quad
			e_{02}=\frac{a_{20}+b_{21}}{b_{10}},\\
			e_{21}&=a_{12}+\frac{a_{11}b_{11}}{b_{10}}+a_{11}v_{11},\\
			e_{12}&=a_{12}v_{02}+\frac{(a_{20}+b_{11})b_{10}+a_{21}b_{10}+2a_{11}b_{20}+a_{20}b_{11}+b_{11}^2 + 2b_{10}b_{12}}{b_{10}^2},\\
			e_{03}&=\frac{2(a_{20}+b_{11})b_{10}^2v_{02}+a_{30}+b_{21}+2a_{20}b_{20}+b_{11}b_{20}}{b_{10}^3}.
		\end{align*}
		Then make the following transformation for the system (\ref{BT-1}):
		\begin{align*}
			\begin{cases}
				x_5=x_3 \\
				x_6=x_4-e_{02}x_3x_4
			\end{cases}
		\end{align*}
		We obtain the system:
		\begin{align}
			\label{BT-2}
			\begin{cases}
				\frac{\mathrm{d}x_5}{\mathrm{d}t}=x_6 +e_{02}x_5x_6+ e_{02}^2x_5^2x_6 +O(\|x\|^4)\\
				\frac{\mathrm{d}x_6}{\mathrm{d}t}=e_{11}x_5x_6  +e_{21}x_5^2x_6 + (e_{12}-e_{02}^2)x_5x_6^2 -e_{02} e_{03}x_6^3+O(\|x\|^4)
			\end{cases}
		\end{align}
		
		Perform a time parameter transformation $t=(1 - \frac{e_{12}-e_{02}^2}{2}x_7^2x_8)\tau$ on the system (\ref{BT-2}) and let:
		\begin{align*}
			\begin{cases}
				x_7=x_5\\
				x_8=x_6+e_{02}x_5x_6+ e_{02}^2x_5^2x_6 +O(\|x\|^4)
			\end{cases}
		\end{align*}
		obtain the system:
		\begin{align}
			\begin{cases}
				\frac{\mathrm{d}x_7}{\mathrm{d}\tau}=x_8-\frac{e_{12}}{2}x_7^2x_8 \\
				\frac{\mathrm{d}x_8}{\mathrm{d}\tau}=e_{11}x_7x_8+e_{21}x_7^2x_8+(e_{12}-e_{02}^2)x_7x_8^2-e_{02} e_{03}x_8^3+O(\|x\|^4)
			\end{cases}
		\end{align}
		let
		\begin{align*}
			\begin{cases}
				x_9=x_7\\
				x_{10}=x_8-\frac{e_{12}}{2}x_7^2x_8
			\end{cases}
		\end{align*}
		so
		\begin{align}
			\begin{cases}
				\frac{\mathrm{d}x_9}{\mathrm{d}\tau}=x_{10} \\ 
				\frac{\mathrm{d}x_{10}}{\mathrm{d}\tau}=f_{11}x_9x_{10}+f_{21}x_9^2x_{10}+f_{03}x_{10}^3+O(\|x\|^4)
			\end{cases}
		\end{align}
		wherein$f_{11}=e_{11},f_{21}=e_{21},f_{03}=-e_{02}e_{03}$.
		
	According to the reference \cite{Yann2008Bifurcation}, the equilibrium point $E_1$ at this time is an equilibrium point with a codimension of 3.
	\end{proof}
	
	According to the theorem, compared with the disease-free equilibrium point $E_0$, for the boundary equilibrium point $E_1$, the threshold value of the infection rate $a_2$ of the second type of infectious disease, which changes the stability of the equilibrium point, becomes larger. This indicates that there is still a certain competitive relationship between the two types of infectious diseases. Therefore, when only one infectious disease is prevalent, it is necessary to control the infection ability of the infectious disease that may have a mixed infection with it, so as to avoid the situation getting out of control. For the population with mixed infections, the threshold value corresponding to its infection rate is consistent with the threshold value discussed in the stability analysis of the disease-free equilibrium point. Therefore, when the situation of mixed infection occurs, the infection ability of the population with mixed infections to other populations should be controlled first. Once the infection ability of the population with mixed infections exceeds the threshold value, both the situation without infectious diseases and the situation with only one infectious disease will no longer be stable, making the spread of infectious diseases a more intractable situation.
	
	Analyze the stability of the interior equilibrium point $E_4$ under special circumstances.
	\begin{theorem}\label{stablethmE4}
		Let $a_1 = \frac{a_2}{a_2 - 1}$, then when $\frac{a_2(2 - a_2)}{a_2 - 1}<r_1<\frac{1}{a_2}+\frac{1}{a_2 - 1}$, the endemic equilibrium point $E_4$ is unstable.
	\end{theorem}
	\begin{proof}
		When $a_1=\frac{a_2}{a_2 - 1}$, we have $k_0 = 1$. Then for the interior equilibrium point $E_4$, we have
		\begin{align*}
			I_m = -\frac{v_1v_2}{u_1v_2+u_2v_1}
		\end{align*}
		Substitute it into the expressions of $I_1$ and $I_2$ in terms of $I_m$, we can get
		\begin{align*}
			I_1=-\frac{u_2v_1^2}{u(u_1v_2+u_2v_1)}\quad I_2=-\frac{u_1v_2 ^2}{u(u_1v_2+u_2v_1)}
		\end{align*}
		Then the Jacobian matrix of the system at this point is:
		\begin{align*}
			J = 
			\begin{pmatrix}
				\frac{a_1u_2v_1^2}{u(u_1v_2 + u_2v_1)} & \frac{(a_1 + a_2)u_2v_1^2}{u(u_1v_2 + u_2v_1)} & \frac{(a_1 + k_0)u_2v_1^2}{u(u_1v_2 + u_2v_1)} \\
				\frac{(a_1 + a_2)u_1v_2^2}{u(u_1v_2 + u_2v_1)} & \frac{a_2u_1v_2^2}{u(u_1v_2 + u_2v_1)} & \frac{(a_2 + k_0)u_1v_2^2}{u(u_1v_2 + u_2v_1)} \\
				-\frac{(a_1 + a_2)u_1v_2^2}{u(u_1v_2 + u_2v_1)} & -\frac{(a_1 + a_2)u_2v_1^2}{u(u_1v_2 + u_2v_1)} & \frac{2v_1v_2u}{(u_1v_2 + u_2v_1)u}
			\end{pmatrix}
		\end{align*}
		Then calculate the corresponding characteristic polynomial:
		\begin{align*}
			&|\lambda E - J| \\
			=&
			\begin{vmatrix}
				\lambda - \frac{a_1u_2v_1^2}{u(u_1v_2 + u_2v_1)} & -\frac{(a_1 + a_2)u_2v_1^2}{u(u_1v_2 + u_2v_1)} & -\frac{(a_1 + k_0)u_2v_1^2}{u(u_1v_2 + u_2v_1)} \\
				-\frac{(a_1 + a_2)u_1v_2^2}{u(u_1v_2 + u_2v_1)} & \lambda - \frac{a_2u_1v_2^2}{u(u_1v_2 + u_2v_1)} & -\frac{(a_2 + k_0)u_1v_2^2}{u(u_1v_2 + u_2v_1)} \\
				\frac{(a_1 + a_2)u_1v_2^2}{u(u_1v_2 + u_2v_1)} & \frac{(a_1 + a_2)u_2v_1^2}{u(u_1v_2 + u_2v_1)} & \lambda - \frac{2v_1v_2u}{u(u_1v_2 + u_2v_1)}
			\end{vmatrix}\\
			=&p_3\lambda^3+p_2\lambda^2+p_1\lambda+p_0 = 0.
		\end{align*}
		where:
		\begin{align*}
			p_3&= 1 ,\quad \quad
			p_2=-\frac{1}{u(u_1v_2 + u_2v_1)}(a_1u_2v_1^2 + a_2u_1v_2^2+2v_1v_2u),\\
			p_1&=\frac{1}{u^2(u_1v_2 + u_2v_1)}\left(2k_0(a_1 + a_2)u_1u_2v_1^2v_2^2+2u_1v_2(a_1u_2v_1^2 + a_2u_1v_2^2)\right),\\
			p_0&=\frac{1}{u^3(u_1v_2 + u_2v_1)^3}\left(u_1^2v_2^2v_1^4(a_1 + a_2)u_1+u_1^2v_2^4v_1^2(a_1 + a_2)u_2+2u_1^2u_2v_1^3v_2^3\right).
		\end{align*}
		Then construct the corresponding Routh table as follows:
		\begin{align*}
			\begin{matrix}
				1&  p_1 &0   \\
				p_2 & p_0 &  0 \\
				\mu_2 & 0 &   \\
				\nu_2 & 0 &  
			\end{matrix}
		\end{align*}
		Correspondingly:
		\begin{align*}
			\mu_2=&\frac{u_1^2u_2v_1^2v_2^4(a_1 + a_2)(u_1 + 2a_2k_0)+u_1u_2^2v_1^4v_2^2(a_1 + a_2)(u_2 + 2a_1k_0)}{u^2(u_1v_2 + u_2v_1)^2(a_1u_2v_1^2+a_2u_1v_2^2 + 2uv_1v_2)}
			\\&+\frac{2u_1u_2u v_1^3v_2^3(u + 2k_0(a_1 + a_2)+2a_1a_2)+2a_1^2uu_2^2v_1^5v_2+2a_2^2uu_1^2v_1v_2^5}{u^2(u_1v_2 + u_2v_1)^2(a_1u_2v_1^2+a_2u_1v_2^2 + 2uv_1v_2)}
			\\&+\frac{4a_1u^2u_2v_1^4v_2^2+4a_2u^2u_1v_1^2v_2^4}{u^2(u_1v_2 + u_2v_1)^2(a_1u_2v_1^2+a_2u_1v_2^2 + 2uv_1v_2)}; \\
			\nu_2=&\frac{1}{u^3(u_1v_2 + u_2v_1)^3}\left(u_1^2v_2^2v_1^4(a_1 + a_2)u_1+u_1^2v_2^4v_1^2(a_1 + a_2)u_2+2u_1^2u_2v_1^3v_2^3\right).
		\end{align*}
		Then by observing the Routh table, we have $1>0, p_2>0, \mu_2>0, \nu_2<0$; Therefore, at this time, the system has one eigenvalue in the right half-plane, that is, this equilibrium point is unstable.
	\end{proof}

	\section{Bifurcation Analysis}
	From the discussion on the stability of the equilibrium points in the previous subsection, it can be concluded that the system may undergo a saddle-node bifurcation of codimension 1 or a Bogdanov-Takens bifurcation of codimension 3. In this subsection, these bifurcation cases will be analyzed.
	\begin{theorem}
		When $v_1 = 0$, the system undergoes a saddle-node bifurcation at the endemic equilibrium point $E_4$.
	\end{theorem}
	\begin{proof}
		For the equilibrium point $E_4$, when $v_1 = 0$, the Jacobian matrix of the system at this point is as follows:
		\begin{align*}
			J_4 = 
			\begin{pmatrix}
				0 & 0 & 0 \\
				\frac{(a_1 + a_2)v_2}{u} & \frac{a_2v_2}{u} & \frac{(a_2 + k)v_2}{u} \\
				-\frac{(a_1 + a_2)v_2}{u} & 0 & k-1
			\end{pmatrix}
		\end{align*}
		Then obviously, the matrix has eigenvalues $\lambda_{1}=0$, $\lambda_{2}=\frac{a_2v_2}{u}$, $\lambda_{3}=k - 1$ at this time. Next, we use the projection method to find the normal form of the bifurcation at this time.
		
		After translating the equilibrium point to the origin through the coordinate transformation $x_1 = I_1, x_2 = I_2+\frac{v_2}{u}, x_3 = I_m$, the system can be expressed in the following form:
		\begin{align}
			\label{foldSys}
			\dot{X}=J_4X+F(X)
		\end{align}
		$X=(x_1,x_2,x_3)^T$,and:
		\begin{align*}
			F(X)&=\begin{pmatrix}
				-a_1x_1^2-(a_1+a_2)x_1x_2-(a_1+k)x_1x_3\\
				- a_2x_2^2-(a_1 + a_2)x_1x_2-(a_2 + k)x_2x_3\\
				(a_1 + a_2)x_1x_2 - kx_3^2
			\end{pmatrix} \\
			&=\frac{1}{2}B(X,X)+\frac{1}{6}C(X,X,X)+O(\left \| X \right \| ^4),
		\end{align*}
		By solving the eigenvector equations $J_4q = 0$ and $J_4^Tp = 0$, we can obtain:
		\begin{align*}
			q = 
			\begin{pmatrix}
				1\\
				-\frac{(a_1 + a_2)[(k - 1)u+(a_2 + k)v_2]}{a_2u(k-1)}\\
				\frac{(a_1 + a_2)v_2}{u(k-1)}
			\end{pmatrix}
			,\quad
			p = \begin{pmatrix}
				1\\
				0\\
				0
			\end{pmatrix}
		\end{align*}
		Then, by the Center Manifold Theorem \cite{carr2012applications}, for $X\in\mathbb{R}^3$, it can be decomposed as $x = wq + y$, where $wq\in T^{c}$ and $y\in T^{sc}$. According to the vectors $p$ and $q$ obtained from the above solution, combined with the formulas for solving $w$ and $y$:
		\begin{align*}
			\begin{cases}
				w =\left \langle p,x \right \rangle \\
				y = x-\left \langle p,x \right \rangle
			\end{cases}
		\end{align*}
		$w = x_1$,$ y =\begin{pmatrix}
			0\\
			x_2+\frac{(a_1+a_2)[(k-1)u+(a_2+k)v_2]}{a_2u}x_1 \\
			x_3-\frac{(a_1+a_2)v_2}{u} x_1
		\end{pmatrix}:=\begin{pmatrix}
			0 \\
			y_1 \\
			y_2
		\end{pmatrix}$
		
		Then, according to the center manifold theorem, the saddle-node bifurcation has the following form:
		\begin{align*}
			\dot{w} = aw^2+bw^3+O(w^4),
		\end{align*}
		wherein
		\begin{align*}
			a&=\frac{1}{2}\left \langle p,B(q,q) \right \rangle, \quad
			b=\frac{1}{6}\left \langle p,C(q,q,q)-3B(q,A^{INV}a) \right \rangle 
		\end{align*}
		Therefore, it can be obtained that the system is equivalent to the following system near the origin:
		\begin{align*}
			\frac{\mathrm{d}w }{\mathrm{d}t} = f_2w^2+f_3w^3+O(\left | w \right | ^4),
		\end{align*}
		$f_2=\frac{1}{2}\left(\frac{(a_1+a_2)^2[(k-1)u+(a_2+k)v_2]}{a_2u(k-1)}
		-\frac{(a_1+k)(a_1+a_2)v_2}{u(k-1)}
		-a_1\right) \ne 0 $;Then let $t = f_2\tau$, and we obtain the following system: 
		\begin{align}
			\label{fold}
			\frac{\mathrm{d}w }{\mathrm{d}\tau} = w^2+\frac{f_3}{f_2}w^3+O(\left | w \right | ^4),
		\end{align}
		Then, according to the literature \cite{guckenheimer2013nonlinear}, \eqref{fold} has the following universal unfolding:
		\begin{align*}
			\frac{\mathrm{d}w}{\mathrm{d}t}=\varepsilon+w^2+O(|w|^3).
		\end{align*}
		So we can get that system(\ref{foldSys}) undergoes a codim-1 saddle-node bifurcation at equilibrium $E_4$.
	\end{proof}
	
	\begin{theorem}\label{BT3E4}
		For the endemic equilibrium point $E_4$ of the endemic disease, when the parameters $k = 1$ and $v_2=0$, the system undergoes a Bogdanov - Takens bifurcation of codimension 3.
	\end{theorem}
	\begin{proof}
		For the equilibrium point $E_4$, when $k = 1$ and $v_2 = 0$, the corresponding Jacobian matrix is as follows:
		\begin{align*}
			J_4(k=1,v_2=0)= \begin{pmatrix}
				\frac{a_1v_1}{u}  & \frac{(a_1+a_2)v_1}{u}  & \frac{(a_1+1)v_1}{u} \\
				0& 0 & 0\\
				0& -\frac{(a_1+a_2)v_1}{u} & 0
			\end{pmatrix}
		\end{align*}
		Then it is obvious that the characteristic matrix has two zero eigenvalues. Consider the following system which translates the equilibrium point to the origin: 
		\begin{align}
			\begin{cases}
				\label{BTE4}
				\frac{\mathrm{d}x}{\mathrm{d}t}&=\frac{a_1v_1}{u}x+\frac{(a_1+a_2)v_1}{u}y+\frac{(a_1+1)v_1}{u}z-(a_1+a_2)xy-(a_1+1)xz \\
				\frac{\mathrm{d}y}{\mathrm{d}t}&=-(a_1+a_2)xy-a_2y^2-(a_2+1)yz \\
				\frac{\mathrm{d}z}{\mathrm{d}t}&=-\frac{(a_1+a_2)v_1}{u}y+(a_1+a_2)xy-z^2 
			\end{cases}
		\end{align}
		
		According to the center manifold theorem, assume that there is a center manifold in the following form:
		\begin{align}
			\label{centerE4}
			x=c_{10}y+c_{01}z+c_{20}y^2+c_{11}yz+c_{02}z^2+O(||(y,z)||^2)
		\end{align}
		Differentiating both sides of the center manifold (\ref{centerE4}) with respect to $t$, we can obtain:
		\begin{align*}
			\frac{\mathrm{d}x}{\mathrm{d}t}&=c_{10}\frac{\mathrm{d}y}{\mathrm{d}t}+c_{01}\frac{\mathrm{d}z}{\mathrm{d}t}+2c_{20}y\frac{\mathrm{d}y}{\mathrm{d}t}+c_{11}y\frac{\mathrm{d}z}{\mathrm{d}t}+c_{11}\frac{\mathrm{d}y}{\mathrm{d}t}z + 2c_{02}z\frac{\mathrm{d}z}{\mathrm{d}t}+\cdots\\
			&=c_{10}\left(-(a_1 + a_2)(c_{10}y + c_{01}z + c_{20}y^2 + c_{11}yz + c_{02}z^2+\cdots)y - a_2y^2 - (a_2 + 1)yz\right)\\
			&+c_{01}\left(-\frac{(a_1 + a_2)v_1}{u}y+(a_1 + a_2)(c_{10}y + c_{01}z + c_{20}y^2 + c_{11}yz + c_{02}z^2+\cdots)y - z^2\right)\\
			&+2c_{20}y\left(-(a_1 + a_2)(c_{10}y + c_{01}z + c_{20}y^2 + c_{11}yz + c_{02}z^2+\cdots)y - a_2y^2 - (a_2 + 1)yz\right)\\
			&+c_{01}z\left(-(a_1 + a_2)(c_{10}y + c_{01}z + c_{20}y^2 + c_{11}yz + c_{02}z^2+\cdots)y - a_2y^2 - (a_2 + 1)yz\right)\\
			&+c_{11}y\left(-\frac{(a_1 + a_2)v_1}{u}y+(a_1 + a_2)(c_{10}y + c_{01}z + c_{20}y^2 + c_{11}yz + c_{02}z^2+\cdots)y - z^2\right)\\
			&+2c_{02}z\left(-\frac{(a_1 + a_2)v_1}{u}y+(a_1 + a_2)(c_{10}y + c_{01}z + c_{20}y^2 + c_{11}yz + c_{02}z^2+\cdots)y - z^2\right)\\
			&+\cdots \\
			&=-\frac{(a_1+a_2)v_1c_{01}}{u}y +y^2\left(-(a_1 + a_2)c_{10}^2 - a_2c_{10}+(a_1 + a_2)c_{10}c_{01}-\frac{(a_1 + a_2)v_1c_{11}}{u}\right)\\
			&+yz\left(-(a_1 + a_2)c_{10}c_{01}-(a_2 + 1)c_{10}+(a_1 + a_2)c_{01}^2-\frac{2(a_1 + a_2)v_1c_{02}}{u}\right)-c_{01}z^2\\
			&+\cdots
		\end{align*}
		Meanwhile, substituting the center manifold (\ref{centerE4}) into the differential equation of the variable $x$ with respect to $t$, we can obtain:
		\begin{align*}
			\frac{\mathrm{d}x}{\mathrm{d}t}=&\frac{a_1v_1}{u}(c_{00}y + c_{01}z + c_{20}y^2 + c_{11}yz + c_{02}z^2+\cdots)\\
			&+(a_1 + a_2)\frac{v_1}{u}y+\frac{(a_1 + 1)v_1}{u}z\\
			&-(a_1 + a_2)(c_{00}y + c_{01}z + c_{20}y^2 + c_{11}yz + c_{02}z^2+\cdots)y\\
			&-(a_1 + 1)(c_{00}y + c_{01}z + c_{20}y^2 + c_{11}yz + c_{02}z^2+\cdots)z\\
			=&\left[\frac{a_1v_1}{u}c_{00}+(a_1 + a_2)\frac{v_1}{u}\right]y+\left(\frac{a_1v_1}{u}c_{01}+\frac{(a_1 + 1)v_1}{u}\right)z\\
			&+y^2\left(\frac{a_1v_1}{u}c_{20}-(a_1 + a_2)c_{00}\right)\\
			&+yz\left(\frac{a_1v_1}{u}c_{11}-(a_1 + a_2)c_{01}-(a_1 + 1)c_{00}\right)\\
			&+z^2\left(\frac{a_1v_1}{u}c_{02}-(a_1 + 1)c_{01}\right)+\cdots
		\end{align*}
		By comparing the coefficients, we can obtain:
		\begin{align*}
			c_{10}&=\frac{a_1+a_2}{a_1^2}, \quad c_{01}=-\frac{a_1+1}{a_1},  \\
			c_{02} &= -\frac{(a_1+1)u}{a_1v_1}, \quad c_{11} = \frac{(a_1+a_2)u}{a_1^4v_1}(2a_1^3+5a_1^2+2a_1+a_2),\\
			c_{20}&=\frac{u}{a_1^5v_1}(2a_1^5+5a_1^4+2a_1^3+10a_1^2a_2+4a_1^4a_2+12a_1^3a_2+2a_1a_2^2+2a_1a_2+a_2^2). 
		\end{align*}
		Denote$(a_1,a_2,k) = (a_1+\varepsilon_1,\frac{1}{r_1}(a_1+a_1^2-a_1r_1+\varepsilon_1)+\varepsilon_2,1+\varepsilon_3)$,and substitute the obtained center manifold (\ref{centerE4}) back into the system (\ref{BTE4}), then we can get the following system:
		\begin{align*}
			\begin{cases}
				\frac{\mathrm{d}y}{\mathrm{d}t}~~=&a_{10}y + a_{01}z + a_{20}y^2 + a_{11}yz + a_{02}z^2 + a_{30}y^3 + a_{21}y^2z + a_{12}yz^2 
				\\&+ a_{03}z^3 + O(\|(y,z)\|^3) \\
				\frac{\mathrm{d}z}{\mathrm{d}t}~~=&b_{10}y + b_{01}z + b_{20}y^2 + b_{11}yz + b_{02}z^2 + b_{30}y^3 + b_{21}y^2z + b_{12}yz^2 
				\\&+ b_{03}z^3 + O(\|(y,z)\|^3)
			\end{cases}	
		\end{align*}
		wherein:
		\begin{align*}
			a_{10} &= O(\varepsilon^2),a_{01} = O(\varepsilon^2),a_{02} = O(\varepsilon^2),a_{03} = O(\varepsilon^2),\\
			a_{20}&=-\frac{1}{r_1}(a_1 + a_1^2 - a_1r_1)-\frac{c_{10}}{r_1}(a_1 + a_1^2)-\frac{1 + c_{10}}{r_1}\varepsilon_1-(1 + c_{10})\varepsilon_2 + O(\varepsilon^2),\\
			a_{11}&=-\frac{1}{r_1}(a_1 + a_1^2 - a_1r_1) - 1-\frac{c_{01}}{r_1}(a_1 + a_1^2)-\frac{1 + c_{01}}{r_1}\varepsilon_1-(1 + c_{01})\varepsilon_2-\varepsilon_3 + O(\varepsilon^2),\\
			a_{30}&=-\frac{(a_1 + a_1^2)c_{20}}{r_1}-\frac{c_{20}}{r_1}\varepsilon_1 - c_{20}\varepsilon_2+ O(\varepsilon^2),
			\\a_{21}&=-\frac{(a_1 + a_1^2)c_{11}}{r_1}-\frac{c_{11}}{r_1}\varepsilon_1 - c_{11}\varepsilon_2+ O(\varepsilon^2)\\
			a_{12}&=-\frac{(a_1 + a_1^2)c_{02}}{r_1}-\frac{c_{02}}{r_1}\varepsilon_1 - c_{02}\varepsilon_2+ O(\varepsilon^2),\\
			b_{10}&=-\frac{(a_1 + a_2)r_1}{u}+\frac{(a_1 + a_2)u_1}{u u_2}\varepsilon_1+ O(\varepsilon^2),b_{01}=\frac{2\varepsilon_1}{u_2}+ O(\varepsilon^2),\\
			b_{20}&=\frac{(a_1 + a_1^2)c_{10}}{r_1}+\frac{c_{10}}{r_1}\varepsilon_1 + c_{10}\varepsilon_2,b_{11}=\frac{(a_1 + a_1^2)c_{01}}{r_1}+\frac{c_{01}}{r_1}\varepsilon_1 + c_{01}\varepsilon_2,
			\\b_{02}&=-\varepsilon_3 + O(\varepsilon^2),
			b_{30}=\frac{(a_1 + a_1^2)c_{20}}{r_1}+\frac{c_{20}}{r_1}\varepsilon_1 + c_{20}\varepsilon_2,
			\\b_{21}&=\frac{(a_1 + a_1^2)c_{11}}{r_1}+\frac{c_{11}}{r_1}\varepsilon_1 + c_{11}\varepsilon_2,b_{12}=\frac{(a_1 + a_1^2)c_{02}}{r_1}+\frac{c_{02}}{r_1}\varepsilon_1 + c_{02}\varepsilon_2,b_{03}= O(\varepsilon^2)
		\end{align*}
		make a variable substitution:
		\begin{align*}
			\begin{cases}
				x_1=z\\
				x_2=b_{10}y
			\end{cases}
		\end{align*}
		then we can obtain:
		\begin{equation}
			\begin{cases}
				\begin{split}
					\frac{\mathrm{d}x_1}{\mathrm{d}t} =~~& b_{01}x_1 + x_2 + b_{02}x_1^2+\frac{b_{11}}{b_{10}}x_1x_2 + \frac{b_{20}}{b_{10}^2}x_2^2 
					\\&+ b_{03}x_1^3+\frac{b_{12}}{b_{10}}x_1^2x_2+\frac{b_{21}}{b_{10}^2}x_1x_2^2+\frac{b_{30}}{b_{10}^3}x_2^3 + O(\|x\|^4)\\
					\frac{\mathrm{d}x_2}{\mathrm{d}t} =~~ &b_{10}a_{01}x_1 + a_{10}x_2 + b_{10}a_{02}x_1^2 + a_{11}x_1x_2 + \frac{a_{20}}{b_{10}}x_2^2 
					\\&+ b_{10}a_{03}x_1^3 + a_{12}x_1^2x_2+\frac{a_{21}}{b_{10}}x_1x_2^2+\frac{a_{30}}{b_{10}^2}x_2^3 + O(\|x\|^4)
				\end{split}
			\end{cases}
		\end{equation}
		Transform the system into a nonlinear oscillator through the following transformation: 
		\begin{align*}
			\begin{cases}
				x_3~~=&x_1\\
				x_4~~=&b_{01}x_1 + x_2 + b_{02}x_1^2+\frac{b_{11}}{b_{10}}x_1x_2 + \frac{b_{20}}{b_{10}^2}x_2^2 + b_{03}x_1^3
				\\&+\frac{b_{12}}{b_{10}}x_1^2x_2+\frac{b_{21}}{b_{10}^2}x_1x_2^2+\frac{b_{30}}{b_{10}^3}x_2^3 + O(\|x\|^4)
			\end{cases}
		\end{align*}
		then the transformed system can be obtained:
		\begin{align*}
			\begin{cases}
				\frac{\mathrm{d}x_3}{\mathrm{d}t}~~=&x_4 \\
				\frac{\mathrm{d}x_4}{\mathrm{d}t}~~=&d_{10}x_3+d_{01}x_4+d_{20}x_3^2+d_{11}x_3x_4+d_{02}x_4^2+d_{30}x_3^3+d_{21}x_3^2x_4+d_{12}x_3x_4^2
				\\&+d_{03}x_4^3+O(||x||^4)
			\end{cases}
		\end{align*}
		\begin{align*}
			d_{10} =& O(\varepsilon^2),d_{01} = b_{01}+O(\varepsilon^2),d_{20} = -\frac{b_{01}b_{11}}{b_{10}} - b_{01}a_{11}+O(\varepsilon^2), \\
			d_{11} =& b_{01}v_{11}+2b_{02}+ a_{11}+\frac{b_{01}b_{11}}{b_{10}} -\frac{2a_{20}b_{01}}{b_{10}}  +O(\varepsilon^2), \\
			d_{02} =& b_{01}v_{02}+\frac{b_{20}+(a_{20}+b_{11})b_{10}}{b_{10}^2}+O(\varepsilon^2), d_{30} = a_{11}v_{20}-a_{12}b_{01}-\frac{a_{11}b_{01}b_{11}}{b_{10}}+O(\varepsilon^2), \\
			d_{21} =& b_{01}v_{21}+a_{11}v_{11}+a_{12}+ 2b_{02}v_{11}+\frac{1}{b_{10}}(5b_{01}b_{12}+b_{01}b_{11}v_{02}-2a_{20}b_{01}v_{11}+2a_{20}v_{20} \\
			&+3b_{02}b_{11}+a_{11}b_{11})+\frac{b_{01}b_{21}}{b_{10}^2}+O(\varepsilon^2) \\
			d_{12}=& b_{01}v_{12}+a_{11}v_{11}+2b_{02}v_{02}+\frac{1}{b_{10}}(b_{01}b_{11}v_{02}-2a_{20}b_{01}v_{12}+2a_{20}v_{11}+a_{21}+b_{01}b_{11}v_{02}+2b_{12})\\
			&+\frac{1}{b_{10}^2}(2b_{20}v_{11}-2b_{01}b_{20}v_{02}+b_{21}-3a_{30}b_{01}+2b_{02}b_{20}+b_{11}^2+a_{20}b_{11}+2a_{11}b_{20}+b_{01}b_{21})\\
			&-\frac{3b_{01}b_{11}b_{20}}{b_{10}^3}\\
			d_{03}=&\frac{(a_{20}+b_{11})2v_{02}}{b_{10}}+\frac{1}{b_{10}^2}(a_{30}+2b_{01}b_{20}v_{02}+b_{21})+\frac{1}{b_{10}^3}(b_{20}b_{11}+2a_{20}b_{20})
		\end{align*}
		In order to eliminate $x_4^2$, make the following transformation:
		\begin{align*}
			\begin{cases}
				x_5=x_3 \\
				x_6=x_4-d_{02}x_3x_4
			\end{cases}
		\end{align*}
		obtain the system:
		\begin{align}
			\label{BT56}
			\begin{cases}
				\frac{\mathrm{d}x_5}{\mathrm{d}t}~~=&x_6 + d_{02}x_5x_6 + d_{02}^2x_5^2x_6+ O(\|x\|^4)\\
				\frac{\mathrm{d}x_6}{\mathrm{d}t}~~=&d_{10}x_5 + d_{01}x_6 + (d_{20}-d_{10}d_{02})x_5^2 + d_{11}x_5x_6+(d_{30}-d_{20}d_{02})x_5^3  \\&+d_{21}x_5^2x_6+(d_{12}-d_{02}^2)x_5x_6^2 - d_{20}d_{03}x_6^3 + O(\|x\|^4)
			\end{cases}
		\end{align}
		Then, perform the following transformation on the system (\ref{BT56}):
		\begin{align*}
			\begin{cases}
				x_7=x_5 \\
				x_8=x_6+ d_{02}x_5x_6 + d_{02}^2x_5^2x_6+ O(\|x\|^4)
			\end{cases}
		\end{align*}
		obtain the system:
		\begin{align}
			\label{BT78}
			\begin{cases}
				\frac{\mathrm{d}x_7}{\mathrm{d}t}~~ =& x_8 \\
				\frac{\mathrm{d}x_8}{\mathrm{d}t}~~ =& e_{00}x_7 + e_{01}x_8 + e_{20}x_7^2 + e_{11}x_7x_8
				\\&+ e_{30}x_7^3 + e_{21}x_7^2x_8+ e_{12}x_7x_8^2 + e_{03}x_8^3+O(\|x\|^4)
			\end{cases}
		\end{align}
		where:
		\begin{align*}
			e_{10}&=d_{10},e_{01}=d_{01},e_{20}=d_{20}-d_{10}d_{02},\\
			e_{11}&=d_{11}-d_{01}d_{02},e_{30}=d_{30}-d_{20}d_{02},\\
			e_{21}&=d_{21}-d_{11}d_{02},e_{12}=d_{12}-d_{02}^2,e_{03}=-d_{20}d_{02}.
		\end{align*}
		In order to eliminate the term $x_7x_8^2$ in the system (\ref{BT78}), make a time parameter transformation $t=(1 - \frac{e_{12}}{2}x_7^2x_8)\tau$, and the following system can be obtained:
		\begin{align*}
			\begin{cases}
				\frac{\mathrm{d}x_7}{\mathrm{d}\tau} ~~=& x_8 - \frac{e_{12}}{2}x_7^2x_8\\
				\frac{\mathrm{d}x_8}{\mathrm{d}\tau}~~ =& e_{10}x_7 + e_{01}x_8 + e_{20}x_7^2 + e_{11}x_7x_8 \\
				&+(e_{30} - \frac{e_{10}e_{12}}{2})x_7^3+ (e_{21} - \frac{e_{01}e_{12}}{2})x_7^2x_8+ e_{12}x_7x_8^2 + e_{03}x_8^3+O(\|x\|^4)
			\end{cases}
		\end{align*}
		Then let:
		\begin{align*}
			\begin{cases}
				x_9=x_7 \\
				x_{10}=x_8-\frac{e_{12}}{2}x_7^2x_8
			\end{cases}
		\end{align*}
		the following system is obtained:
		\begin{align}
			\label{BT910}
			\begin{cases}
				\frac{\mathrm{d}x_9}{\mathrm{d}\tau}~~=& x_{10} \\
				\frac{\mathrm{d}x_{10}}{\mathrm{d}\tau}~~=& f_{10}x_9 + f_{01}x_{10} + f_{20}x_9^2 + f_{11}x_9x_{10}
				\\&+ f_{30}x_9^3 + f_{21}x_9^2x_{10}+  f_{03}x_{10}^3+O(\|x\|^4)
			\end{cases}
		\end{align}
		Perform a scaling transformation on the system (\ref{BT910}):
		\begin{align*}
			x_9=\frac{1}{\sqrt[4]{f_{03}f_{21}} }y_1,\quad x_{10}=\frac{f_{21}^\frac{1}{4} }{f_{03}^\frac{3}{4}}y_2 ,\quad t=\sqrt{\frac{f_{21}}{f_{03}}}\tau.
		\end{align*}
		obtain the system:
		\begin{align}
			\begin{cases}
				\frac{\mathrm{d}y_1}{\mathrm{d}t}=y_2 \\
				\frac{\mathrm{d}y_2}{\mathrm{d}t}=\eta_1 y_1+\eta_2 y_2+\eta_3 y_1^2 +g_{11}y_1y_2+g_{30}y_1^3+y_1^2y_2+y_2^3+O(||y||^4)
			\end{cases}
		\end{align}
		where$\eta_1=\frac{f_{10}f_{03}}{f_{21}},\eta_2=\frac{f_{01}\sqrt{f_{03}}}{\sqrt{f_{21}}},\eta_3=\frac{f_{20}f_{03}^\frac{3}{4}}{f_{21}^\frac{5}{4}},g_{11}=\frac{f_{11}f_{03}^\frac{1}{4}}{f_{21}^\frac{3}{4}},g_{30}=\frac{f_{30}f_{03}^\frac{1}{2}}{f_{21}^\frac{3}{2}}$,:
		\begin{align*}
			\frac{\partial (\eta_1,\eta_2,\eta_3)}{\partial (\varepsilon_1,\varepsilon_2,\varepsilon_3) } \ne 0
		\end{align*}
		Therefore, it can be concluded that when there are small perturbations to the parameters $(\eta_1,\eta_2,\eta_3)$, the system (\ref{BTE4}) undergoes a B-T bifurcation of codimension 3 with $a_1,a_2,k$ as the bifurcation parameters. 
	\end{proof}
	
	\section{Numerical simulation}
	According to Theorem \ref{stablethmE1}, when the invasion reproduction number $R_2 < 1$, the equilibrium point $E_1$ is a stable node. Select the parameters $a_1 = 3$, $r_1 = 2$, $a_2 = 0.5$, $k = 0.5$, and present the graphs of the relationships between the densities of three types of infected populations and time:
	\begin{figure}[H]
		\centering
		\includegraphics[scale=0.6]{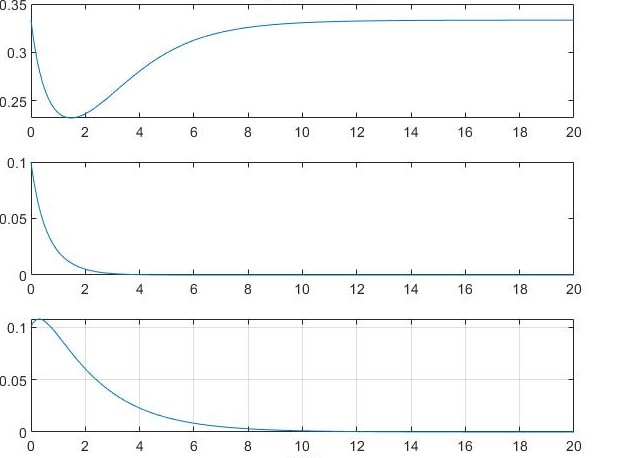}
		\caption{The density change curves of three types of populations when $R_2 < 1$}
		\label{figure}
	\end{figure}
	From the graph, it can be seen that for the given initial values $(0.35, 0.1, 0.1)$, the density of the population $I_1$ infected with the first type of infectious disease will first decrease, then increase, and finally approach a stable value; the density of the population $I_2$ infected with the second type of infectious disease will keep decreasing and finally remain at 0; the density of the population $I_3$ infected with both infectious diseases will first increase slightly, then decrease, and finally remain at 0. This trend of change verifies that when the invasion reproduction number $R_2 < 1$, the equilibrium point $E_1$ is a stable equilibrium point.
	
	Next, conduct a numerical simulation of the bifurcation at the equilibrium point $E_4$. According to Theorem \ref{exist4E4E5}, when $k = k_0\leq1$ and $v_1 < 0$, $v_2 < 0$, there exists an interior equilibrium point $E_4$. And according to Theorem \ref{BT3E4}, when $k = 1$ and $v_2 = 0$, the system undergoes a B - T bifurcation of codimension 3 at this interior equilibrium point. Therefore, when the parameter values $a_1 = 2$, $a_2 = 2$, $r_1 = 1.5$ are selected, the system (\ref{BTE4}) has the following phase diagram:
	\begin{figure}[H]
		\label{E-4B-T3}
		\centering
		\includegraphics[scale=0.5]{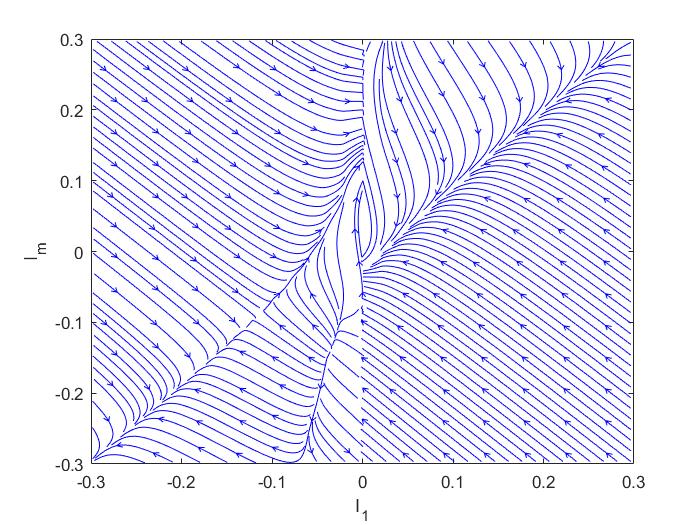}
		\caption{The phase diagram of the system near $E_4$ when $k = 1$ and $v_2=0$}
	\end{figure}
	
	\section{Conclusions}
	Based on the traditional infectious disease model, this paper constructs an infectious disease co-infection model under complex circumstances, and analyzes the relevant contents of this model through the qualitative theory and bifurcation theory in the dynamical system. Firstly, by means of analysis, the parameter conditions for the existence of the disease-free equilibrium point, boundary equilibrium points and interior equilibrium points are given. Secondly, the stability of each boundary equilibrium point is analyzed correspondingly, and the stability of the interior equilibrium point $E_4$ is presented under special conditions. Finally, at the interior equilibrium point $E_4$, when the parameter $v_1 = 0$, the system undergoes a saddle-node bifurcation at this point; when taking $(a_1,a_2,k)$ as parameters and making small perturbations to these parameters, the system will experience a B-T bifurcation of codimension 3. 
	
	\newpage
	\bibliographystyle{plain}
	
\end{document}